\documentclass[12pt]{article}

\usepackage{amssymb,amsfonts,amsthm}\usepackage{graphicx}
\usepackage{subfigure}

\usepackage{amsmath}
\usepackage{bm, bbm}             
\usepackage[mathscr]{eucal}
\usepackage{color}
\usepackage{cancel}
\usepackage{enumerate}
\usepackage{psfrag}
\usepackage{appendix}
\usepackage{tensor}

\newcommand{\R}{\mathbb{R}}

\newcommand{\overbar}[1]{\mkern 1.5mu\overline{\mkern-1.5mu#1\mkern-1.5mu}\mkern 1.5mu}

\newtheorem{proposition}{Proposition}

\theoremstyle{definition}

\newtheorem*{remark}{Remark}

\usepackage{fullpage}

\numberwithin{equation}{section}

\title{\Huge{Nonlinear diffusion in\\ relativistic kinetic theory}} \author
 {
       Simone Calogero  \\
       {\small Department of Mathematical Sciences}  \\
       {\small Chalmers University of Technology}  \\
       {\small Gothenburg, Sweden} \\
       }

\date{}
\begin{document}
\maketitle
\begin{abstract}
A nonlinear Lorentz invariant kinetic diffusion equation is introduced, which is consistent with the conservation laws of particles number, energy and momentum. The equilibrium solution converges to the Maxwellian density in the Newtonian limit, but it is not given by the Jüttner distribution commonly employed in relativistic kinetic theory. 
The nonlinear kinetic diffusion equation on a general Lorentzian manifold is consistent with the contracted Bianchi identities and therefore can be coupled to the Einstein equations of general relativity.
\end{abstract}
\section{Introduction}
Consider a system of a large number of particles and let $f(t,{\bm x},{\bm v})$ be the number density of particles with position ${\bm x}\in\R^3$ and velocity ${\bm v}\in\R^3$ at time $t\in\R$. If the only forces acting on the particles are due to random collisions with the molecules of a background medium in thermal equilibrium, the function $f$ solves the linear kinetic diffusion equation
\begin{equation}\label{linkin}
\partial_t f+{\bm v}\cdot\nabla_{\bm x} f=\sigma\Delta_{\bm v} f,
\end{equation}
where $\sigma>0$ is a constant with physical  dimension $[\mathrm{velocity}]^2/[\mathrm{time}]$ (diffusion constant).  Equation~\eqref{linkin} is known by different names in the literature, e.g., kinetic Fokker-Planck or Kramers equation (without friction). 
Furthermore,~\eqref{linkin} is the (forward) Kolmogorov PDE associated to the stochastic process 
\begin{align*}
&{\bm X}(s)={\bm x}+\int_s^t {\bm V}(\tau)\,d\tau,\\
 &{\bm V}(s)={\bm v}+\sqrt{2\sigma}({\bm B}(t)-{\bm B}(s)),
\end{align*}
where ${\bm B}(t)$ is a three-dimensional Brownian motion.
The stochastic process $({\bm X}(s), {\bm V}(s))$ represents the microscopic (random) state of the particles. 
The solution for $t>0$ of~\eqref{linkin} with initial data $f(0,{\bm x}, {\bm v})=f_0({\bm x}, {\bm v})$ is given by Feynman-Kac's formula
\[
f(t,{\bm x},{\bm v})=\mathbb{E}[f_0({\bm X}(0),{\bm V}(0))]=\int_{\R^6}K(t,{\bm x},{\bm v},{\bm y}, {\bm w})f_0({\bm y}, {\bm w})\,d{\bm w}\,d{\bm y},
\]
where $K(t,{\bm x},{\bm v},\cdot)$ is the density of the random variable $({\bm X}(0),{\bm V}(0))$---i.e., the fundamental solution solution of~\eqref{linkin}---which is known explicitly~\cite{Ch}.
When coupled to a suitable mean field---e.g., the self-induced electric field in the case of charged particles---the linear kinetic diffusion equation~\eqref{linkin} provides the foundation for numerous phenomenological models in physics and engineering.

One notable feature of~\eqref{linkin} is that it is invariant by Galilean transformations---a fundamental property for Newtonian mechanical systems. Yet, a questionable feature of~\eqref{linkin} is that it violates the conservation laws of energy and momentum. In Section~\ref{newtonian} we present a well-known {\it nonlinear} version of~\eqref{linkin},  which is Galilean invariant and which preserves the particles number, energy and momentum. In the spatially homogeneous case, the nonlinear model becomes the {\it linear} one found in~\cite{KL} and, after a proper translation in the velocity variable, it reduces to the spatially homogeneous Kramers equation.

Following the pioneering work by Dudley~\cite{D}, several relativistic versions of~\eqref{linkin} have been proposed in the literature; see~\cite{DH, H} and the review~\cite{DH2}. In~\cite{AC} it was shown that the only Lorentz invariant equation among the proposed ones is
\begin{equation}\label{RFP}
\partial_tf+{\bm {\widehat{u}}}\cdot\nabla_{\bm x}f=\sigma\,\nabla_{\bm u}\cdot \left(D\nabla_{\bm u}f\right),
\end{equation}
where ${\bm u}\in\R^3$ is the three-velocity of the particles, and 
\[
{\bm {\widehat{u}}}=\frac{c{\bm u}}{\sqrt{c^2+|{\bm u}|^2}},\quad D=\frac{c^2\mathbb{I}+{\bm u}\otimes{\bm u}}{u^0c}
\]
are, respectively, the relativistic velocity and the diffusion matrix. As~\eqref{linkin},~\eqref{RFP} also violates the conservation laws of energy and momentum.
In Section~\ref{relativistic} we introduce a {\it nonlinear} version of~\eqref{RFP}, which is Lorentz invariant and which preserves the particles number, energy and momentum. In the Newtonian limit $c\to\infty$, the new relativistic model converges to the Newtonian one from Section~\ref{newtonian}. However, in contrast to the latter, the relativistic model does not become linear in the spatially homogeneous case.

Remarkably, the equilibrium distribution $G_\mathrm{eq}$ of the nonlinear Lorentz invariant diffusion equation derived in Section~\ref{relativistic} is not the standard Jüttner distribution universally employed in relativistic kinetic theory. Instead, for a particles distribution with zero average momentum, it is given by
\[
G_\mathrm{eq}({\bm u})=\frac{c^{-3}Z(\lambda)}{(1+c^{-2}|{\bm u}|^2)^{\lambda}},
\]
where $\lambda$ is a dimensionless constant that depends on the speed of light and $Z(\lambda)$ is a dimensionless normalization factor so that $\|G_\mathrm{eq}\|_{L^1}=1$ (partition function).
The probability density $G_\mathrm{eq}({\bm u})$ converges to the standard Maxwellian distribution in the Newtonian limit.

As the conservation law of energy-momentum is a necessary constraint for the existence of solutions to the Einstein equations of general relativity, the kinetic diffusion equation~\eqref{RFP} is incompatible with Einstein's theory of gravity. In~\cite{C} this inconsistency was resolved  
by postulating the existence in spacetime of a cosmological scalar field $\phi$ that transfers energy-momentum to the particles. 
(It was shown in~\cite{CV} that $\phi$ could be identified with the dark energy component of the universe in observational cosmology.)
In Section~\ref{genrelsec} we propose another solution to this  problem; specifically, we show that the nonlinear kinetic diffusion equation introduced in Section~\ref{relativistic} is compatible with the contracted Bianchi identities in Lorentzian geometry and thus can be coupled to the Einstein equations without the need to introduce additional matter fields in spacetime.

\section{Nonlinear kinetic diffusion in Newtonian mechanics}\label{newtonian}
Let $f(t,{\bm x},{\bm v})\geq 0$ be the number density in state space for an isolated system of point particles with mass $m> 0$.
The number density $n_f=n_f(t,{\bm x})$ and the number current density ${\bm j}_f={\bm j}_f(t,{\bm x})$ of particles in space are
\[
n_f=\int  f\,d{\bm v},\quad {\bm j}_f= \int  {\bm v}\, f\,d{\bm v};
\]
integrals with unspecified domain are extended over $\R^3$. 
The kinetic energy density $e_f=e_f(t,{\bm x})$ and the momentum density ${\bm p}_f={\bm p}_f(t,{\bm x})$ in space are
\[
e_f=\frac{m}{2}\int  |{\bm v}|^2f\,d{\bm v},\quad {\bm p}_f=m\int  {\bm v} f\,dv=m\,{\bm j}_f.
\]
The local conservation laws of particles number, energy and momentum read
\begin{equation}\label{consloc}
\partial_t n_f =-\nabla_{\bm x}\cdot {\bm j}_f,\quad \partial_te_f=-\nabla_{\bm x}\cdot {\bm q}_f,\quad \partial_t {\bm p}_f =-\nabla_{\bm x}\cdot S_f,
\end{equation}
where
\[
{\bm q}_f=\frac{m}{2}\int  |{\bm v}|^2 {\bm v} f\,d{\bm v},\quad S_{f}=m\int  {\bm v}\otimes {\bm v} f\,d{\bm v} 
\]
are, respectively, the energy current density and the momentum current density (or stress tensor) in space; we assume that no chemical or other form of reaction occurs among the particles. 

Suppose that $f$ satisfies an equation of the form
\begin{equation}\label{genFP}
\partial_tf+{\bm v}\cdot\nabla_{\bm x}f=\Phi_f\Delta_{\bm v}f+\nabla_{\bm v}\cdot({\bm W}_f f),
\end{equation}
where the diffusion scalar field $\Phi_f=\Phi_f(t,{\bm x})$ and the drift vector field ${\bm W}_f={\bm W}_f(t,{\bm x}, {\bm v})$ depend on $f$. 
We ask for conditions on $\Phi_f$, ${\bm W}_f$ such that~\eqref{genFP} is invariant under the Galilean transformation
\[
(\bar{t},\bar{{\bm x}}, \bar{{\bm v}})=
(t, {\bm x}-{\bm \xi}t, {\bm v}-{\bm \xi}),
\]
where ${\bm \xi}\in\R^3$ is an arbitrary constant vector. For this question to be meaningful, we need to specify how the kinetic density $\bar{f}$ of the new Galilean observer is related to $f$. 
It is natural to assume that
\begin{equation}\label{transf}
\bar{f}(t,{\bm x},{\bm v})=f(\bar{t},\bar{{\bm x}}, \bar{{\bm v}}).
\end{equation}
The transformation~\eqref{transf} ensures that Galilean observers agree on the total number of particles at the given absolute time $t=\bar{t}$:
\[
\int \int \bar{f}(t,{\bm x},{\bm v})\,d{\bm v}\,d{\bm x} = \int \int f(\bar{t},\bar{{\bm x}}, \bar{{\bm v}})\,d{\bm v}\,d{\bm x}=\int \int f(t,{\bm x},{\bm v})\,d{\bm v}\,d{\bm x},
\]
where we used that the measure $d{\bm v}\,d{\bm x}$ is invariant by Galilean transformations. The previous identity holds regardless of whether the total number of particles is constant.
It is now straightforward to verify that~\eqref{genFP} is Galilean invariant (i.e., $\bar{f}(t,{\bm x},{\bm v})$ is a solution) if 
\begin{equation}\label{trans}
\Phi_{\bar{f}}(t,{\bm x})=\Phi_f(\bar{t},\bar{{\bm x}}),\quad
{\bm W}_{\bar{f}}(t,{\bm x},{\bm v})={\bm W}_f(\bar{t},\bar{{\bm x}},\bar{{\bm v}}).
\end{equation} 

The next purpose is to find $\Phi_f$, ${\bm W}_f$ satisfying~\eqref{trans} and such that~\eqref{genFP} implies the validity of the conservation laws~\eqref{consloc}. 
The conservation law of particles number holds for all choices of $\Phi_f, {\bm W}_f$; in fact, integrating by parts in the velocity variable and assuming that all boundary terms at infinity vanish, we find 
\[
\partial_tn_f = \int  \partial_tf\,d{\bm v}=-\int {\bm v}\cdot\nabla_{\bm x}f\,d{\bm v}+\Phi_f\int \Delta_{\bm v}f\,d{\bm v}+\int \nabla_{\bm v}\cdot({\bm W}_f f)\,d{\bm v}=-\nabla_{\bm x}\cdot{\bm j}_f.
\]
Similarly, we obtain
\[
\partial_te_f=-\nabla_{\bm x}\cdot {\bm q}_f+m\left(3n_f\,\Phi_f -\int {\bm W}_f\cdot{\bm v} f\,d{\bm v}\right),\quad \partial_t {\bm p}_f =-\nabla_{\bm x}\cdot S_f-m\int {\bm W}_ff\,d{\bm v}.
\]
Thus, the conservation laws~\eqref{consloc} hold for all solutions $f$ if we set 
\begin{equation}\label{wsigma}
\Phi_f=\frac{\mu}{3n_f}\left(\frac{2e_f}{m}-\frac{|{\bm j}_f|^2}{n_f}\right),\quad  {\bm W}_f=\mu\left({\bm v}-\frac{{\bm j}_f}{n_f}\right),
\end{equation}
where $\mu$ is a constant with physical dimension $[\mathrm{time}]^{-1}$. 
By analogy with the linear Kramers equation in kinetic theory~\cite{K}, we shall refer to $\mu$ as the drag (or viscosity) coefficient. 

\begin{remark}
\textnormal{More generally, the drag coefficient could depend on a spacetime invariant constructed from the kinetic density $f$. For instance, in~\cite[Eq.~(20)] {Villani} it is defined as $\mu\sim (n_f)^\alpha$, where $\alpha\in[0,1]$.}
\end{remark}

It is straightforward to verify that the diffusion scalar field $\Phi_f$ and the drift vector field ${\bm W}_f$ given by~\eqref{wsigma} satisfy~\eqref{trans}. For instance,
\begin{align*}
{\bm W}_{\bar{f}}(t,{\bm x},{\bm v})&=\mu\left({\bm v}-\frac{{\bm j}_{\bar{f}}}{n_{\bar{f}}}\right)=\mu\left({\bm v}-\frac{\displaystyle{\int} {\bm y} \bar{f}(t,{\bm x}, {\bm y})\,d{\bm y}}{\displaystyle{\int} \bar{f}(t,{\bm x}, {\bm y})\,d{\bm y}}\right)=\mu \left({\bm v}-\frac{\displaystyle{\int} {\bm y} f(t,{\bm x}-{\bm \xi}t, {\bm y}-{\bm \xi})\,d{\bm y}}{\displaystyle{\int} f(t,{\bm x}-{\bm \xi}t, {\bm y}-{\bm \xi})\,d{\bm y}}\right)\\
&=\mu \left({\bm v}-\frac{\displaystyle{\int} ({\bm z}+{\bm \xi}) f(t,{\bm x}-{\bm \xi}t, {\bm z})\,d{\bm z}}{\displaystyle{\int} f(t,{\bm x}-{\bm\xi}t, {\bm z})\,d{\bm z}}\right)=\mu \left({\bm v}-{\bm \xi}-\frac{\displaystyle{\int} {\bm z} f(t,{\bm x}-{\bm\xi}t, {\bm z})\,d{\bm z}}{\displaystyle{\int} f(t,{\bm x}-{\bm\xi}t, {\bm z})\,d{\bm z}}\right)\\
&={\bm W}_f(\bar{t},\bar{{\bm x}}, \bar{{\bm v}}),
\end{align*}
and similarly one proves that $\Phi_{\bar{f}}(t,{\bm x})=\Phi_f(\bar{t},\bar{{\bm x}})$. In conclusion, when $\Phi_f$ and ${\bm W}_f$ are given by~\eqref{wsigma}, the nonlinear kinetic diffusion equation~\eqref{genFP} is Galilean variant and implies the conservation laws~\eqref{consloc}.
Moreover, since
\[
\Phi_f=\frac{\mu}{3n_f}\int \left|{\bm v}-(n_f)^{-1}{\bm j}_f\right|^2f\,d{\bm v},
\]
then diffusion scalar field $\Phi_f$ is non-negative.

\subsection*{Spatially homogeneous solutions}
A particle system is spatially homogeneous if there exists a non-negative function $g=g(t,{\bm v})$ such that 
\begin{equation}\label{shf}
f(t,{\bm x},{\bm v})=\frac{g(t,{\bm v})}{\mathrm{Vol}(\Omega)},
\end{equation}
where $\Omega\subset\R^3$ is the (finite) region occupied by the particles. 
Galilean observers agree on the form~\eqref{shf} of the kinetic density, because
\[
\bar{f}(t,{\bm x},{\bm v})=f(\bar{t},\bar{\bm x},\bar{\bm v})=f(t,\bar{\bm x},\bar{\bm v})=\frac{g(t,\bar{\bm v})}{\mathrm{Vol}(\Omega)}:=\frac{\bar{g}(t,\bm v)}{\mathrm{Vol}(\Omega)},
\]
where we used that the volume of the region occupied by the particles is Galilean invariant. 

Replacing~\eqref{shf} in~\eqref{genFP}, we find the following equation on $g$:
\begin{equation}\label{genFPsh}
\partial_tg=\Phi_{g}\Delta_{\bm v} g+\nabla_{\bm v}\cdot({\bm W}_{g}g),
\end{equation}
where $\Phi_g$ and ${\bm W}_g$ are given by~\eqref{wsigma} with
\[
n_g=\int g\,d{\bm v},\quad {\bm j}_g=\int {\bm v}g\,d{\bm v},\quad e_g=\frac{m}{2}\int |{\bm v}|^2g\,d{\bm v}.
\]
Since $n_g$, ${\bm j}_g$, $e_g$ are constant for solutions of~\eqref{genFPsh}, then ${\bm W}_g$ is now a time independent vector field and $\Phi_g$ is a positive constant. It follows in particular that~\eqref{genFPsh} is a {\it linear} equation. Denoting
\[
{\bm j}_g={\bm \omega},\quad \Phi_g=\sigma,
\]
and introducing the probability density 
\[
G(t,{\bm v})=g(t,{\bm v})/n_g,
\]
we may rewrite~\eqref{genFPsh} as
\begin{subequations}\label{genFPshTEMP}
\begin{align}
&\partial_t G=\sigma\Delta_{\bm v} G+\nabla_{\bm v}\cdot(\mu({\bm v}-{\bm \omega})G)\\
&\int  G\,d{\bm v}=1,\quad \int {\bm v}\,G\,d{\bm v}={\bm \omega},\quad \frac{\mu}{3}\int |{\bm v}-{\bm \omega}|^2G\,d{\bm v}=\sigma.
\end{align}
\end{subequations}
\begin{remark}\textnormal{
In~\cite{KL} the authors arrive to the same equation~\eqref{genFPshTEMP} by a different, more fundamental, argument. Specifically, they start from a spatially homogeneous system of $N$ identical particles in a periodic box and with random velocities ${\bm V}=({\bm v_1},{\bm v_2},\dots, {\bm v_N})$ taking value on the manifold $\mathbb{M}$ of constant energy and momentum. The probability density $F^{(N)}({\bm V}, t)$ of the particle system is assumed to satisfy the heat equation on $\mathbb{M}$ and~\eqref{genFPshTEMP} is derived in the limit $N\to\infty$. }
\end{remark}
The time-independent solution of~\eqref{genFPshTEMP} is given by the non-central Maxwellian 
\begin{equation}\label{noncentralmax}
\mathcal{M}^{({\bm \omega})}=\mathcal{M}({\bm v}-{\bm \omega}),
\end{equation}
where
\begin{equation}\label{maxwell}
\mathcal{M}({\bm v}):=\mathcal{M}^{({\bm 0})}({\bm v})=
\left(\frac{\mu}{2\pi \sigma}\right)^{3/2}\exp\left(-\frac{\mu|{\bm v}|^2}{2\sigma}\right)
\end{equation}
is the central Maxwellian density. 
It is well-known that the solutions of~\eqref{genFPshTEMP} converge exponentially fast to the equilibrium~\eqref{noncentralmax}; see e.g.~\cite{AMU} for a proof of this property using the entropy method. It is an interesting open question whether exponential convergence to equilibrium holds for the nonlinear kinetic diffusion equation~\eqref{genFP}. The local asymptotic stability of the Maxwellian distribution~\eqref{maxwell} in the spatially inhomogeneous case has been proved in~\cite{LY}.

\section{Nonlinear kinetic diffusion in special relativity}\label{relativistic}
Consider now a collection of a large number of relativistic particles. Let $x=(x^0,{\bm x})\in\R^4$ be a system of Cartesian coordinates in Minkowski space---i.e., a coordinate system in which the Minkoswki metric $\eta$ has components $\eta_{\alpha\beta}=\mathrm{diag}(-1,1,1,1)$---and denote by $u=(u^0,{\bm u})\in\R^4$ the four-dimensional velocity variable. The state space of each individual particle is therefore the seven-dimensional manifold 
\[
\Pi=\R^4\times\mathcal{U}, \quad \mathcal{U}=\{u\in\R^4: \ u^0=\sqrt{c^2+|{\bm u}|^2}\},
\]
where $c$ is the speed of light. 
The kinetic density $f\geq0$ of relativistic particles is defined on $\Pi$ and thus can be written as a function of $(x,{\bm u})$. However, in order to derive a Lorentz invariant equation on $f$ that generalizes~\eqref{genFP}, it is convenient to start from a density $\mathbbm{f}=\mathbbm{f}(x,u)$ ``off-shell'', i.e., to which the condition $u^0=\sqrt{c^2+|{\bm u}|^2}$ has not yet been imposed. The relation between $f$ and $\mathbbm{f}$ is therefore
\[
f(x, {\bm u})=\mathbbm{f}(x, \sqrt{c^2+|{\bm u}|^2}, {\bm u}).
\]
As a starting point we postulate the following equation on $\mathbbm{f}$:
\begin{equation}\label{eqf*}
u^\alpha\partial_{x^\alpha}\mathbbm{f}=\Phi_\mathbbm{f} \, \eta^{\alpha\beta}\partial_{u^\alpha}\partial_{u^\beta} \mathbbm{f}+\partial_{u^\alpha} (K_\mathbbm{f}^\alpha \, \mathbbm{f}),
\end{equation}
where $\Phi_\mathbbm{f} =\Phi_\mathbbm{f} (x)$ is a scalar field with physical dimension $[\mathrm{velocity}]^2/[\mathrm{time}]$ and $K_\mathbbm{f} =K_\mathbbm{f} (x,u)$ is a four-dimensional vector field with physical dimension $[\mathrm{velocity}]/[\mathrm{time}]$, both of which depend on $\mathbbm{f}$. Consider now the Lorentz transformation
\begin{equation}\label{inv4d}
\bar{x}=\Lambda({\bm\xi}) x, \quad \bar{u}=\Lambda({\bm\xi}) u,\quad \bar{\mathbbm{f}}(x,u)= \mathbbm{f}(\bar{x},\bar{u}),
\end{equation}
where ${\bm \xi}\in\R^3$ and $\Lambda({\bm\xi})$ is the $4\times 4$ matrix
\[
\Lambda({\bm\xi}) =\frac{1}{c}\left(\begin{array}{cc} \displaystyle{\xi^0} & \displaystyle{-{\bm \xi}} \\ \displaystyle{-{\bm \xi}} & \displaystyle{c\mathbb{I}+\frac{{\bm \xi}\otimes{\bm \xi}}{\xi^0+c}}\end{array}\right),\quad \xi^0=\sqrt{c^2+|{\bm \xi}|^2}.
\]
Using the identities
\begin{equation}\label{lambdaid}
\Lambda({\bm\xi})^T\eta\, \Lambda({\bm \xi})=\eta,
\quad \Lambda({\bm \xi})^{-1}=\Lambda(-{\bm \xi}),
\end{equation}
we find that~\eqref{eqf*} is invariant under the Lorentz transformation~\eqref{inv4d} (i.e., $\bar{\mathbbm{f}}(x,u)$ is a solution)
provided $\Phi_\mathbbm{f} ,K_\mathbbm{f} $ transform according to
\begin{equation}\label{transphik}
\Phi_{\bar{\mathbbm{f}}}(x)=\Phi_\mathbbm{f} (\bar{x}),\quad {K}_{\bar{\mathbbm{f}}}(x,u) =\Lambda(-{\bm \xi}) K_\mathbbm{f} (\bar{x},\bar{u}).
\end{equation}
In the following discussion we simplify the notation by writing $(\Phi,K)$ instead of $(\Phi_\mathbbm{f},K_\mathbbm{f})$. 
Assume that $K $ is tangent to $\Pi$; that is, $K ^\alpha u_\alpha=0$. Equivalently,
\begin{equation}\label{K0}
K^0=\frac{{\bm K} \cdot{\bm u}}{u^0},
\end{equation}
where ${\bm K} =(K ^1,K ^2,K ^3)\in\R^3$ denotes the spatial part of $K $. By~\eqref{K0}, the projection of~\eqref{eqf*} on $\Pi$ gives the following equation on $f$:
\begin{equation}\label{eqf}
\sqrt{c^2+|{\bm u}|^2}\partial_{x^0}f+{\bm u}\cdot\nabla_{{\bm x}}f=\Phi \Delta_{{\bm u}}^{(h)}f+\nabla_{{\bm u}}^{(h)}\cdot({\bm K}f),
\end{equation} 
where  $\nabla_{{\bm u}}^{(h)}\cdot $, $\Delta_{{\bm u}}^{(h)}$ denote, respectively, the divergence and Laplace-Beltrami operators associated to the Riemannian metric $h$  induced by $\eta$ on the hyperboloid $u^0=\sqrt{c^2+|{\bm u}|^2}$. In the coordinates $u^1,u^2,u^3$ we have
\[
h_{ij}=\delta_{ij}-\frac{u_iu_j}{c^2+|{\bm u}|^2}
\]
and 
\[
\nabla_{{\bm u}}^{(h)}\cdot({\bm K}f)=\frac{1}{\sqrt{|h|}}\partial_{u^i}\left(\sqrt{|h|}\,K^i f\right),\quad
\Delta_{{\bm u}}^{(h)}f=\frac{1}{\sqrt{|h|}}\partial_{u^i}\left(\sqrt{|h|}(h^{-1})^{ij}\partial_{u^j}f\right),
\]
where $|h|=\det(h_{ij})=c^2(c^2+|{\bm u}|^2)^{-1}$ and $(h^{-1})^{ij}$ is the matrix inverse of $h_{ij}$, that is 
\begin{equation}\label{inverseh}
(h^{-1})^{ij}=\delta^{ij}+\frac{u^iu^j}{c^2}.
\end{equation}
It follows that~\eqref{eqf} is
\begin{equation}\label{eqf2}
u^\alpha\partial_{x^\alpha}f=u^0\partial_{u^i}\left(\Phi\frac{(h^{-1})^{ij}}{u^0}\partial_{u^j}f+\frac{K^i}{u^0}f\right),
\end{equation}
where from now on it is understood that $u^0=\sqrt{c^2+|{\bm u}|^2}$.
The next goal is to find $\Phi,K$ satisfying~\eqref{transphik} and such that~\eqref{eqf} implies the validity of the local conservation laws of particles number, energy and momentum. 

The local conservation law of particles number can be expressed as
\begin{equation}\label{consn}
\partial_{x^\alpha}N^\alpha=0,
\end{equation}
where
\begin{equation}\label{N}
N^\alpha(x) =\int u^\alpha f(x,{\bm u})\frac{d{\bm u}}{u^0}
\end{equation}
is the four-dimensional number current density. It is easy to see that~\eqref{consn} holds for any choice of $\Phi,K$. Indeed we have
\[
\partial_{x^\alpha}N^\alpha
=\int u^\alpha \partial_{x^\alpha} f\,\frac{d{\bm u}}{u^0}=\int \partial_{u^i}\left[\frac{\Phi}{u^0}\tensor{(h^{-1})}{^i_j}\partial_{u^j}f+\frac{K^i}{u^0}f\right]d{\bm u}.
\]
Hence, $\partial_{x^\alpha}N^\alpha=0$, provided $f$ and $|\nabla_{{\bm u}}f|$ decay to zero sufficiently fast as $|{\bm u}|\to\infty$.

We now turn the attention to the conservation law of energy-momentum, which we write in the form
\begin{equation}\label{consT}
\partial_{x^\beta}T^{\alpha\beta}=0,
\end{equation}
where
\begin{equation}\label{T}
T^{\alpha\beta}(x)=mc\int u^\alpha u^\beta f(x,{\bm u})\frac{d{\bm u}}{u^0}
\end{equation}
is the stress-energy(-momentum) tensor.  Computing the left hand side of~\eqref{consT}, using~\eqref{eqf2} and integration by parts, we obtain 
\begin{align*}
&\partial_{x^\beta}T^{0\beta}=mc\left[\frac{3\Phi}{c^2}\int f\,d{\bm u}-\int \frac{{\bm K}\cdot{\bm u}}{(u^0)^2}\,f\,d{\bm u}\right],\\
&\partial_{x^\beta}T^{i\beta}=mc\left[\frac{3\Phi}{c^2}\int \frac{u^i}{u^0}\,f\,d{\bm u}-\int \frac{K^i}{u^0}\,f\,d{\bm u}\right].
\end{align*}
The local conservation law of energy $\partial_{x^\beta}T^{0\beta}=0$ entails
\begin{equation}\label{sigma1}
\Phi=\frac{c^2}{3}\left(\int f\,d{\bm u}\right)^{-1}\int \frac{{\bm K}\cdot{\bm u}}{(u^0)^2}\,f\,d{\bm u}
\end{equation}
and thus the local conservation law of momentum $\partial_{x^\beta}T^{i\beta}=0$ gives
\[
\frac{N^i}{N^0}\int \frac{{\bm K}\cdot{\bm u}}{(u^0)^2}\,f\,d{\bm u}-\int \frac{K^i}{u^0}\,f\,d{\bm u}=0.
\]
We rewrite the last equation as
\begin{equation}\label{tempimp}
\int \frac{\tensor{A}{^i_j}K^j}{u^0}\,f\,d{\bm u}=0,\quad \tensor{A}{^i_j}=\tensor{\delta}{^i_j}-\frac{N^i}{N^0} \frac{u_j}{u^0}.
\end{equation}
The simplest choice of $K$ that is consistent with~\eqref{tempimp} and (a posteriori) with~\eqref{transphik} is
\[
K^i=\mu \tensor{(A^{-1})}{^i_j}\left(u^j-\frac{u^0}{N^0}N^j\right),
\]
where $\mu$ is the drag coefficient; the constant $\mu$ is introduced in the previous equation for dimensional reasons.
Using that the matrix inverse of $\tensor{A}{^i_j}$ is 
\[
\tensor{(A^{-1})}{^i_j}=\tensor{\delta}{^i_j}-\frac{N^iu_j}{N^\alpha u_\alpha},
\]
we find
\begin{equation}\label{ki}
K^i=\mu\left(u^i+\frac{c^2 N^i}{N^\alpha u_\alpha}\right).
\end{equation}
Using~\eqref{K0} we obtain
\[
K^0=\frac{{\bm K}\cdot {\bm u}}{u^0}=\mu\left(u^0+\frac{c^2N^0}{N^\alpha u_\alpha}\right).
\] 
Hence, the vector field $K$ in~\eqref{eqf*} is
\[
K=\mu\left(u+\frac{c^2N}{N^\alpha u_\alpha}\right),
\]
while the diffusion scalar field~\eqref{sigma1} is
\begin{equation}\label{sigma2}
\Phi=\frac{c^2}{3N^0}\int \frac{K^0}{u^0}\,f\,d{\bm u}=\frac{c^2\mu}{3}\left(1+\int \frac{c^2f}{N^\alpha u_\alpha}\frac{d{\bm u}}{u^0}\right).
\end{equation}
In conclusion, our proposal for the relativistic generalization of~\eqref{genFP} is Equation~\eqref{eqf2} with~\eqref{ki} and~\eqref{sigma2} substituted in. It remains to show that the found equation is indeed Lorentz invariant.
\begin{proposition}\label{inv4dprop}
$\Phi$, $K$ satisfy the transformation laws~\eqref{transphik}; in particular,~\eqref{eqf*}, and thus also~\eqref{eqf2}, are invariant under Lorentz transformations.
\end{proposition}
\begin{proof}
To prove the claim, we rewrite $\Phi$ and $K$ in terms of $\mathbbm{f}$ (instead of $f$). The definition~\eqref{N} of $N$ is equivalent to 
\[
N(x)=\int_\mathcal{U} u\,\mathbbm{f}(x, u)\,du,
\]
where $\mathcal{U}=\{u\in \R^4:u^0=\sqrt{c^2+|{\bm u}|^2}\}$. By the change of variable $u=\bar{v}=\Lambda({\bm\xi}) v$ we obtain
\[
N(\bar{x})=\int_\mathcal{U} \Lambda({\bm\xi}) v\,\mathbbm{f}(\bar{x}, \bar{v})\,dv=\Lambda({\bm\xi})\int_\mathcal{U} v\,\bar{\mathbbm{f}}(x,v)\,dv:=\Lambda({\bm\xi}) \overbar{N}(x),
\] 
where we used that $\det\Lambda({\bm\xi})=1$, and that $\mathcal{U}$ is invariant under Lorentz transformations.
Hence, using~\eqref{lambdaid},
\begin{align*}
K_\mathbbm{f} (\bar{x},\bar{u})&=\mu \left(\bar{u}+\frac{c^2N(\bar{x})}{N(\bar{x})^\alpha \bar{u}_\alpha}\right)=\mu\Lambda({\bm\xi})\left(u+\frac{c^2\overbar{N}(x)}{\overbar{N}(x)^\alpha u_\alpha}\right)=\Lambda({\bm\xi})K_{\bar{\mathbbm{f}}}(x,u)\\
&\Rightarrow K_{\bar{\mathbbm{f}}}(x,u)=\Lambda(-{\bm \xi})K_\mathbbm{f} (\bar{x},\bar{u}).
\end{align*}
The claim on $\Phi$ is proved likewise:
\begin{align*}
\Phi_\mathbbm{f} (\bar{x})&=\frac{c^2\mu}{3}\left(1+\int_\mathcal{U} \frac{c^2\mathbbm{f}(\bar{x},u)}{N(\bar{x})^\alpha u_\alpha}\,du\right)=\frac{c^2\mu}{3}\left(1+\int_\mathcal{U} \frac{c^2\mathbbm{f}(\bar{x},\bar{v})}{N(\bar{x})^\alpha \bar{v}_\alpha}\,d v\right)\\
&=\frac{c^2\mu}{3}\left(1+\int_\mathcal{U} \frac{c^2\bar{\mathbbm{f}}(x,v)}{\overbar{N}(x)^\alpha v_\alpha}\,dv\right)=\Phi_{\bar{\mathbbm{f}}}(x).
\end{align*}
\end{proof}

\subsection{3+1 formulation and Newtonian limit}
We start this section by rewriting~\eqref{eqf2} in a form similar to the Newtonian equation~\eqref{genFP}.
Let us introduce the time variable $t=x^0/c$. Denote $f(c t,{\bm x},{\bm u})$ simply by $f(t,{\bm x}, {\bm u})$ and define
\[
n_f(t,{\bm x})=\int f(t,{\bm x},{\bm u}) \, d{\bm u},\quad {\bm j}_f(t,{\bm x})=\int  {\bm {\widehat{u}}}f(t,{\bm x},{\bm u})\,d{\bm u},
\]
where 
\[
{\bm {\widehat{u}}}=\frac{c{\bm u}}{u^0},\quad u^0=\sqrt{c^2+|{\bm u}|^2},
\]
is the relativistic velocity. The conservation law~\eqref{consn} of particles number is equivalent to
\[
\partial_t n_f=-\nabla_{\bm x}\cdot{\bm j}_f.
\]
Let us also introduce the non-negative function $\beta_f=\beta_f(t,{\bm x},{\bm u})$ by
\[
\beta_f=\frac{c\,n_fu^0-{\bm j}_f\cdot{\bm u}}{c^2}.
\]
Equivalently,
\begin{equation}\label{epsilon2}
\beta_f(t,{\bm x},{\bm u})=c^{-2}\int f(t,{\bm x},{\bm v})(cu^0-\frac{c{\bm v}}{v^0}\cdot{\bm u})\,d{\bm v}=c^{-1}\int f(t,{\bm x},{\bm v})\frac{|u^\mu v_\mu|}{v^0}\,d{\bm v},
\end{equation}
where $v=(v^0,{\bm v})$, $v^0=\sqrt{c^2+|{\bm v}|^2}$, and where for the second equality we used that the Lorentzian scalar product of timelike vectors is negative.
In terms of these new variables, the relativistic kinetic diffusion equation~\eqref{eqf2} takes the final form
\begin{equation}\label{FPrel}
\partial_tf+{\bm {\widehat{u}}}\cdot\nabla_{\bm x}f=\Phi_f\,\nabla_{\bm u}\cdot\left(D\nabla_{\bm u}f\right)+\nabla_{{\bm u}}\cdot({\bm W}_ff),
\end{equation}
where the drift vector field ${\bm W}_f={\bm W}_f(t,{\bm x}, {\bm u})$ is given by
\[
{\bm W}_f=
\frac{c\mu}{u^0}\left({\bm u}-\frac{{\bm j}_f}{\beta_f}\right),
\]
the diffusion matrix $D=D({\bm u})$ is 
\[
D=\frac{c^2\mathbb{I}+{\bm u}\otimes {\bm u}}{u^0c},
\]
and the diffusion scalar field $\Phi_f=\Phi_f(t,{\bm x})$ is
\begin{equation}\label{phif}
\Phi_f=\frac{c^2\mu}{3}\left(1-\int \frac{cf}{\beta_f}\frac{d{\bm u}}{u^0}\right).
\end{equation}
In the Newtonian limit $c\to\infty$ we have ${\bm {\widehat{u}}}\to {\bm u}$ and $D\to \mathbb{I}$. Moreover,
\[
\frac{1}{\beta_f}=\frac{c^2}{c\,n_fu^0-{\bm j}_f\cdot{\bm u}}=\frac{1}{n_f}+\frac{{\bm j}_f\cdot{\bm u}-n_f|{\bm u}|^2}{c^2n_f^2}+O(c^{-4})
\]
and therefore
\[
{\bm W}_f\to\mu\left({\bm u}-\frac{{\bm j}_f}{n_f}\right),\quad \Phi_f\to\frac{\mu}{3n_f}\left(\int |{\bm u}|^2f\,d{\bm u}-\frac{|{\bm j}_f|^2}{n_f}\right),\quad \text{as $c\to\infty$,}
\]
in agreement with~\eqref{wsigma}.

By straightforward calculations one can show that the fields ${\bm W}_f$, $\Phi_f$ satisfy the identities
\begin{equation}\label{identities}
(i) \int {\bm W}_f f\,d{\bm u}=\frac{3\Phi_f}{c^2}{\bm j}_f,\quad (ii) \int \frac{{\bm W}_f\cdot{\bm u}}{u^0}f\,d{\bm u}=\frac{3\Phi_f}{c}n_f,\quad 
(iii)\ {\bm W}_f=D\nabla_{{\bm u}}\left(c^2\mu\log\beta_f\right).
\end{equation}
The identities $(i)$-$(ii)$ can be used to prove the conservation law~\eqref{consT} of energy-momentum in the form
\[
\partial_t e_f=-\nabla_{\bm x}\cdot{\bm q}_f,\quad \partial_t {\bm p}_f=-\nabla_{\bm x}\cdot S_f,
\]
where
\[
e_f=mc\int u^0f\,d{\bm u},\quad  {\bm q}_f=mc^2\int {\bm u}f\,d{\bm u},\quad
{\bm p}_f={\bm q}_f/c^2,\quad S_f=mc\int \frac{{\bm u}\otimes {\bm u}}{u^0}f\,d{\bm u}.
\]
For instance,
\begin{align*}
\partial_te_f=mc\int u^0\partial_t f\,d{\bm u}&=-mc^2\int {\bm u}\cdot\nabla_{\bm x}f\,d{\bm u}+mc\,\Phi_f\int u^0\nabla_{{\bm u}}(D\nabla_{{\bm u}}f)\,d{\bm u}\\
&\quad +mc\int u^0\nabla_{\bm u}\cdot({\bm W}_ff)\,d{\bm u}\\
&=-\nabla_{\bm x}\cdot{\bm q}_f-mc\,\Phi_f\int\frac{D{\bm u}}{u^0}\cdot\nabla_{{\bm u}}f-mc\int\frac{{\bm u}}{u^0}\cdot{\bm W}_ff\,d{\bm u}\\
&=-\nabla_{\bm x}\cdot{\bm q}_f+mc\left(\frac{3\Phi_f n_f}{c}-\int\frac{{\bm u}}{u^0}\cdot{\bm W}_ff\,d{\bm u}\right),
\end{align*}
where we used that $D{\bm u}=c^{-1}u^0{\bm u}$. Hence, the conservation law of energy $\partial_te_f=-\nabla_{\bm x}\cdot{\bm q}_f$ follows by the identity $(ii)$ in~\eqref{identities}. Similarly, one proves the conservation law of momentum $ \partial_t {\bm p}_f=-\nabla_{\bm x}\cdot S_f$ using the identity $(i)$ in~\eqref{identities}. An important difference with the Newtonian case is that the current density ${\bm j}_f$ and the momentum density ${\bm p}_f$ are no longer parallel vector fields. In particular, ${\bm j}_f$ does not satisfy a conservation law. 

To conclude this section we prove that the diffusion scalar field~\eqref{phif} is non-negative and bounded from above. Specifically, 
\[
0\leq \Phi_f(t,{\bm x})\leq \frac{c^2\mu}{3}.
\]
The upper bound is obvious, as $\beta_f\geq0 $. 
For the lower bound we use that $(u^\mu v_\mu)^2\geq (u^\mu u_\mu) (v^\mu v_\mu)$ holds for all timelike vectors $u,v$, with equality if and only if $u$ and $v$ are parallel; see~\cite[Theorem 1.4.1]{N}. In particular, 
\begin{equation}\label{belowIMP}
|u^\mu v_\mu|\geq c^2\quad \text{for all $u,v\in\mathcal{U}$},
\end{equation}
and $|u^\mu v_\mu|=c^2$ if and only if $u=v$.
By~\eqref{epsilon2} and~\eqref{belowIMP}, 
\[
\beta_f(t,{\bm x},{\bm u})\geq \int c f(t,{\bm x},{\bm v})\,\frac{d{\bm v}}{v^0},
\]
with equality for all $(t,{\bm x}, {\bm u})$ if and only if $f\equiv 0$ 
(because the set of ${\bm v}\in\R^3$ such that $v=u$ has zero Lebesgue measure). 
Hence,
\[
\Phi_f(t,{\bm x})=\frac{c^2\mu}{3}\left(1-\int \frac{cf(t,{\bm x},{\bm u})}{\beta_f(t,{\bm x},{\bm u})}\frac{d{\bm u}}{u^0}\right)\geq \frac{c^2\mu}{3}\left(1-\frac{\displaystyle{\int cf(t,{\bm x},{\bm u})\,\frac{d{\bm u}}{u^0}}}{\displaystyle{\int c f(t,{\bm x},{\bm v})\,\frac{d{\bm v}}{v^0}}}\right)=0.
\]

\subsection{Spatially homogeneous solutions}\label{SHsec}
As the Lorentz transformation of the time variable $t$ introduces a dependence on the space variable ${\bm x}$,  a kinetic density $f$ cannot have the form~\eqref{shf} for all Lorentzian observers. The definition of spatially homogeneous particle system in the relativistic case is that there exists a Lorentzian observer $O$ such that  
\[
f(t,{\bm x}, {\bm u})=g(t,{\bm u})/\mathrm{Vol}(\Omega),
\]
where $\Omega$ is the finite region occupied by the particles (according to $O$). The observer $O$ is defined up to time translations and spatial rotations.

In terms of the probability density 
\[
G(t,{\bm u})=g(t,{\bm u})/n_g,\quad \int G\,d{\bm u}=1,
\]
the diffusion equation~\eqref{FPrel} for a spatially homogeneous particle system in the frame of the observer $O$ becomes
\begin{equation}\label{FPrelHOM2}
\partial_tG=\Phi_G\,\nabla_{\bm u}\cdot\left(D\nabla_{\bm u}G\right)+\nabla_{{\bm u}}\cdot({\bm W}_G\,G),
\end{equation}
where 
$\Phi_G=\Phi_G(t)$, ${\bm W}_G={\bm W}_G(t,{\bm u})$ are given by
\begin{align*}
&\Phi_G=\frac{c^2\mu}{3}\left(1-\int \frac{c\, G}{\beta_G}\frac{d{\bm u}}{u^0}\right),\quad{\bm W}_G=\frac{c\mu}{u^0}\left({\bm u}-\frac{{\bm j}_G}{\beta_G}\right),\\
 &\beta_G=\frac{c u^0-{\bm j}_G\cdot{\bm u}}{c^2},\quad {\bm j}_G=\int  {\bm {\widehat{u}}}\,G\, d{\bm u}.
\end{align*}
\eqref{FPrelHOM2} is the relativistic generalization of~\eqref{genFPshTEMP}. 
However, contrary to~\eqref{genFPshTEMP},~\eqref{FPrelHOM2} is not linear and does not preserve the current density ${\bm j}_G$. 
The average energy $e_G$ and momentum ${\bm p}_G$ of the particles are given by
\[
e_G=mc\int u^0 G\, d{\bm u},\quad {\bm p}_G=m\int {\bm u}\, G\, d{\bm u}
\]
and are preserved by smooth solutions of~\eqref{FPrelHOM2}. 

Next we derive the time independent solution $G_\mathrm{eq}({\bm u})$ of~\eqref{FPrelHOM2} and show that $G_\mathrm{eq}$ converges, as $c\to\infty$, to the (non-central) Maxwellian density; see Proposition~\ref{convMax}. Let $\sigma>0$ and ${\bm y}\in\R^3$ be fixed; we start by looking to the solution $G_\mathrm{eq}({\bm u})$ of~\eqref{FPrelHOM2} such that 
\begin{equation}\label{sigmay}
\Phi_{G_\mathrm{eq}}=\sigma,\quad {\bm j}_{G_{\mathrm{eq}}}={\bm y}.
\end{equation}
By the definition of $\Phi_{G_\mathrm{eq}}$, the first condition in~\eqref{sigmay} is equivalent to
\[
\int \frac{c\, G_\mathrm{eq}}{\beta_G}\frac{d{\bm u}}{u^0}=1-\frac{3\sigma}{c^2\mu}
\]
and therefore we require
\[
\sigma<\frac{c^2\mu}{3}.
\]
Using the identity $(iii)$ in~\eqref{identities}, we find that $G_\mathrm{eq}$ satisfies the equation
\begin{equation}\label{eqeq}
\nabla_{\bm u}\cdot \left[DG_\mathrm{eq}\nabla_{{\bm u}}\left(\log G_\mathrm{eq}+2\lambda\log\left(c^{-1}u^0-c^{-2}{\bm y}\cdot {\bm u}\right)\right)\right]=0,
\end{equation}
where
\[
\lambda=\frac{c^2\mu}{2\sigma}>\frac{3}{2}.
\]
By~\eqref{eqeq}, the equilibrium solution $G_\mathrm{eq}=G_\mathrm{eq}^{(\lambda,{\bm y})}$ is given by
\begin{subequations}\label{eqsol}
\begin{equation}
G_\mathrm{eq}^{(\lambda,{\bm y})}({\bm u})=\frac{c^{-3}Z(\lambda,{\bm y})}{\left(c^{-1}u^0-c^{-2}{\bm y}\cdot {\bm u}\right)^{2\lambda}},
\end{equation}
where $Z(\lambda,{\bm y})$ is a dimensionless normalization factor such that $\|G_\mathrm{eq}\|_{L^1}=1$; that is 
\[
Z(\lambda,{\bm y})=c^3\left(\int\left(\sqrt{1+c^{-2}|{\bm u}|^2}-c^{-2}{\bm y}\cdot {\bm u}\right)^{-2\lambda}d{\bm u}\right)^{-1}.
\] 
Computing the integral we find
\begin{equation}
Z(\lambda,{\bm y})=\frac{(1-|{\bm y}/c|^2)^{\lambda+1/2}\Gamma\left(\lambda\right)}{\pi^{3/2}\Gamma  \left(\lambda-3/2\right)}.
\end{equation}
\end{subequations}
The probability density~\eqref{eqsol} satisfies the constraints~\eqref{sigmay} for all $0<\sigma<c^2\mu/3$ and ${\bm y}\in\R^3$. 
\begin{proposition}\label{convMax}
$G_\mathrm{eq}^{(\lambda,{\bm y})}({\bm u})\to \mathcal{M}({\bm u}-{\bm y})$ as $c\to\infty$, where $\mathcal{M}$ is the Maxwellian density~\eqref{maxwell}.
\end{proposition}
\begin{proof}
Replacing $c=\sqrt{2\sigma\lambda/\mu}$ in~\eqref{eqsol}, and using $\Gamma(\lambda+\alpha)\sim \Gamma(\lambda)\lambda^\alpha$, $\lambda\to\infty$, we find
\[
c^{-3}Z(\lambda,{\bm y})\sim \frac{\left(1-\displaystyle{\frac{\mu|{\bm y}|^2}{2\sigma\lambda}}\right)^{\lambda+1/2}}{(2\pi\sigma/\mu)^{3/2}}\to \left(\frac{\mu}{2\pi\sigma}\right)^{3/2}\exp\left({\displaystyle{-\frac{\mu}{2\sigma}|{\bm y}|^2}}\right),
\]
\[
\left(\sqrt{1+c^{-2}|{\bm u}|^2}-c^{-2}{\bm y}\cdot {\bm u}\right)^{-2\lambda}
\to \exp\left({-\frac{\mu}{2\sigma}(|{\bm u}|^2-2{\bm y}\cdot{\bm u})}\right).
\]
The claim follows.
\end{proof}

The average particles energy $e_{G_\mathrm{eq}^{(\lambda,{\bm y})}}$ and the average momentum ${\bm p}_{G_\mathrm{eq}^{(\lambda,{\bm y})}}$ of the equilibrium distribution~\eqref{eqsol} are finite if and only if
$\lambda>2$, i.e.,
\[
\sigma<\frac{c^2\mu}{4},
\]
in which case they are given by
\begin{align*}
&{\bm p}_{G_\mathrm{eq}^{(\lambda,{\bm y})}}=m\,{\bm \omega},\quad {\bm \omega}=\int {\bm u}\,G_\mathrm{eq}^{(\lambda,{\bm y})}({\bm u})\,d{\bm u}=\frac{2\lambda {\bm y}}{\sqrt{1-|{\bm y}/c|^2}}\,\Theta(\lambda),\\
&e_{G_\mathrm{eq}^{(\lambda,{\bm y})}}=mc\,\varepsilon,\quad \varepsilon=\int u^0G_\mathrm{eq}^{(\lambda,{\bm y})}({\bm u})\,d{\bm u}=\frac{c[(2\lambda-1)+|{\bm y}/c|^2]}{\sqrt{1-|{\bm y}/c|^2}}\,\Theta(\lambda),
\end{align*}
where
\[
\Theta(\lambda)=\frac{\Gamma(\lambda-2)\Gamma(\lambda)}{(2\lambda-1)\Gamma(\lambda-3/2)\Gamma(\lambda-1/2)}.
\]

\section{Nonlinear kinetic diffusion in general relativity}\label{genrelsec}
Let $(M,g)$ be a spacetime, i.e., a four-dimensional time-oriented Lorentzian manifold. The state space of a particle with mass $m>0$ is the seven-dimensional submanifold of the tangent bundle given by
\[
\Pi=\cup_{x\in M}\Pi[x],\quad \Pi[x]=\{u\in T_x M: g_{\mu\nu}u^\mu u^\nu=-1,\ u \text{ future directed}\}
\] 
and thus the kinetic density of particles with position $x\in M$ and four-velocity $u$ is a function
\[
f:\Pi\to [0,\infty).
\]
To keep the discussion in this section as close as possible to the special relativistic case, we introduce an orthonormal frame  $e_{(\mu)}$  such that $e_{(0)}$ is timelike and future pointing. Denoting by $x^\mu=(x^0,{\bm x})$ a (local) coordinates system on $M$ and by $u^\mu=(u^0,{\bm u})$ the components of the particles four-velocity in the frame $e_{(\mu)}$, the state space conditions in the definition of $\Pi[x]$ entail 
\[
u^0=\sqrt{c^2+|{\bm u}|^2}
\]
and $f=f(x,{\bm u})$. Moreover, the components $N^\alpha$ of the number current density and the components $T^{\alpha\beta}$ of the stress energy tensor in the frame $e_{(\mu)}$ are given by~\eqref{N} and~\eqref{T}, respectively.

The analog of~\eqref{eqf2} in the general relativistic case is  
\begin{equation}\label{FPgenrel}
u^\mu e_{(\mu)}^\nu\partial_{x^\nu}f+\tensor{\gamma}{_\mu_\nu}^iu^\mu u^\nu\partial_{u^i}f=u^0\partial_{u^i}\left(\Phi\frac{(h^{-1})^{ij}}{u^0}\partial_{u^j}f+\frac{K^i}{u^0}f\right),
\end{equation}
where $(h^{-1})^{ij}$, $\Phi$, $K^i$ are given by~\eqref{inverseh},~\eqref{ki},~\eqref{sigma2} and $\tensor{\gamma}{_\alpha_\beta^\mu}$ are the Ricci rotation coefficients of the frame $e_{(\mu)}$; that is
\[
\tensor{\gamma}{_\alpha_\beta^\mu}=e_{(\alpha)}^a e_{(\beta)}^b\nabla_ae_b^{(\mu)},
\]
where $e_a^{(\mu)}$ is the co-frame dual to $e_{(\mu)}^a$ and $\nabla_a$ is the Levi-Civita connection (the abstract index notation is employed; see~\cite{Wald}). 
Thus, the only difference with the special relativistic case is the definition of the operator in the left hand side of~\eqref{FPgenrel}, which is now the projection on $\Pi$ of the geodesic flow vector field; see~\cite{C2}. Because of this modification, the local conservation laws of particles number and energy-momentum hold now in the general relativistic form
\[
\nabla_{\alpha}N^\alpha=0,\quad \nabla_\beta T^{\alpha\beta}=0.
\]
As $\nabla_\beta T^{\alpha\beta}=0$, the kinetic diffusion equation~\eqref{FPgenrel} can be coupled to the Einstein equations
\begin{equation}\label{eineq}
R_{\alpha\beta}-\frac{1}{2}g_{\alpha\beta}R=\frac{8\pi G}{c^2}T_{\alpha\beta}.
\end{equation}
In particular, and as opposed to the linear kinetic diffusion equation introduced in~\cite{C}, it is no longer necessary to add a cosmological scalar field to the left hand side to~\eqref{eineq}, or alter the Einstein equations in any other way, because~\eqref{FPgenrel} is now compatible with the (contracted) Bianchi identity  $\nabla_\alpha(R^{\alpha\beta}-\frac{1}{2}g^{\alpha\beta}R)=0$. Applications to cosmology of the new general relativistic diffusion theory will be discussed in a future publication.

\end{document}